\documentclass[journal,letterpaper,twocolumn,twoside,nofonttune]{IEEEtran}

\usepackage{amsmath}
\usepackage{amsfonts}
\usepackage{bbding}   
\usepackage{stmaryrd}
\usepackage{latexsym}
\usepackage{amssymb}
\usepackage{graphics}
\usepackage{graphicx}
\usepackage{colordvi}
\usepackage{comment}
\usepackage{color}
\usepackage{tikz,pgfplots}

\usepackage{mathrsfs} % Define \mathscr, a script font

\usepackage{cite} %\def\citepunct{,\,} \def\citedash{--}
\usepackage{upref}

\usepackage{theorem}

\usepackage{setspace}
\usepackage{algorithm, algorithmic}

\usepackage{multirow}

\usepackage{comment}  % Comments out with \begin{comment} ... \end{comment}

\renewcommand{\markboth}[2]
{\renewcommand{\leftmark}{#1}\renewcommand{\rightmark}{#2}}

\newcommand{\EXP}[1]{\mathsf{E}\!\left[#1\right] }

\usepackage{todonotes}

\theoremstyle{plain} 
\theorembodyfont{\normalfont\slshape}

\newtheorem{thm}{Theorem\hspace{-1pt}} 
\newenvironment{theorem}
{\begin{thm}\hspace*{-1ex}{\bf.}}{\end{thm}}

\newtheorem{lem}[thm]{Lemma\hspace{-.75pt}}
\newenvironment{lemma}{\begin{lem}\hspace*{-1ex}{\bf.}}{\end{lem}}

\newtheorem{prop}[thm]{Proposition$\!$}

\newtheorem{cor}[thm]{Corollary$\!$}

\newtheorem{defn}{Definition$\!$}
\newenvironment{definition}{\begin{defn}\hspace*{-1ex}{\bf.}}{\end{defn}}

\newtheorem{assumption}{Assumption}

\newcounter{enumrom}
\renewcommand{\theenumrom}{(\roman{enumrom})}

\makeatletter
\renewcommand{\@endtheorem}{\endtrivlist}
\makeatother

\makeatletter
\renewcommand{\thefigure}{{\bf \@arabic\c@figure}}
\renewcommand{\fnum@figure}{{\bf Figure}\,\thefigure}
\makeatother

\renewcommand{\QEDclosed}{\mbox{\rule[-1pt]{1.3ex}{1.3ex}}} % 

\makeatletter
\DeclareRobustCommand{\myfrac}[1]{%
  \global\@xp\let\csname#1\@xp\endcsname\csname @@#1\endcsname
  \csname#1\endcsname
}
\makeatother

\newcommand{\cF}{{\cal F}}

\newcommand{\cI}{{\cal I}}

\newcommand{\cL}{{\cal L}}
\newcommand{\cM}{{\cal M}} 
\newcommand{\cN}{{\cal N}}

\newcommand{\cT}{{\cal T}}

\DeclareMathAlphabet{\mathbfsl}{OT1}{ppl}{b}{it} %{OT1}{cmr}{bx}{it}

\newcommand{\beps}{\mathbf{\epsilon}} 

\newcommand{\be}[1]{\begin{equation}\label{#1}}
\newcommand{\ee}{\end{equation}} 
\newcommand{\eq}[1]{(\ref{#1})}

\renewcommand{\leq}{\leqslant}
 
\renewcommand{\geq}{\geqslant}

\renewcommand{\Bbb}{\mathbb}

\newcommand{\R}{{\Bbb R}}

\newcommand{\Tref}[1]{Theo\-rem\,\ref{#1}}

\newcommand{\Lref}[1]{Lem\-ma\,\ref{#1}}
\newcommand{\Cref}[1]{Co\-ro\-lla\-ry\,\ref{#1}}

\newcommand{\bet}{\beta}

\newcommand{\ep}{\varepsilon}
\newcommand{\eps}{\varepsilon}

\newcommand{\gd}{\cT g}

\title{Global Games with Noisy Information Sharing}
\author{
Hessam Mahdavifar, Ahmad Beirami, Behrouz Touri, and Jeff S.\ Shamma
\thanks{The work of B.\ Touri is supported by National Science Foundation under the grant NSF-ECCS 1610003. The work of J.\ S.\ Shamma is supported by funding from King Abdullah University of Science and Technology (KAUST).}
\thanks{H.\ Mahdavifar is with the Department of Electrical and Computer Engineering, University of Michigan, Ann Arbor, MI 48109, USA (email: hessam@umich.edu).}
\thanks{A.\ Beirami is with the School of Engineering and Applied Sciences, Harvard University, Cambridge, MA, USA, and also with the Electrical Engineering and Computer Science Department, Massachusetts Institute of Technology, Cambridge, MA, USA (email: beirami@seas.harvard.edu, beirami@mit.edu).}
\thanks{B.\ Touri is with University of California San Diego, Electrical and Computer Engineering Department, La Jolla  92093, USA (email: btouri@eng.ucsd.edu).}
\thanks{J.\ S.\ Shamma is with King Abdullah University of Science and Technology (KAUST), Computer, Electrical and Mathematical Science and Engineering Division (CEMSE), Thuwal 23955--6900, Saudi Arabia (email: jeff.shamma@kaust.edu.sa).}
\vspace{-.25in}
}

\begin{document}
\maketitle
\begin{abstract}
Global games form a subclass of games with incomplete information where a set of agents decide actions against a regime with an underlying fundamental $\theta$ representing its power. Each agent has access to an independent noisy observation of $\theta$. In order to capture the behavior of agents in a social network of information exchange we assume that agents share their observation in a noisy environment prior to making their decision. We show that global games with noisy sharing of information do not admit an intuitive type of threshold policy which only depends on agents' belief about the underlying $\theta$. This is in contrast to the existing results on the threshold policy for the conventional set-up of global games. Motivated by this result, we investigate the existence of equilibrium strategies in a more general collection of threshold-type policies and show that such equilibrium strategies exist and are
unique if the sharing of information happens over a sufficiently noisy
environment. %To show this result, we establish that if a threshold function is an equilibrium strategy, then it will be a solution to a fixed point equation. Then, we show that for a sufficiently noisy environment, the functional fixed point equation leads to a contraction mapping, and hence, its iterations converge to a unique continuous threshold policy.

\end{abstract}

%==============================================================================%
%                                                                              %
%    1. INTRODUCTION                                                           %
%                                                                              %
%==============================================================================%
%%%%%%%%%%%%%%%%%%%%%%%%%%%%%%%%%%%%%%%%%%%%%%%%%%%%%%%%%%%%%%%%%%%%%%%%
\section{Introduction}

Games of incomplete information are central to modeling of socio-economic behaviors in social networks. In games with incomplete information, the information that is shared among the agents is not symmetric and often each agent has access to limited information about the game's parameters, which may be correlated with the information that is available to the other agents. This incompleteness of information is often represented by considering the payoff as a function of some random variable whose exact value is not known to the agents, which is called the underlying economic or political fundamental in the society. 

A subclass of games with incomplete information is the class of global games, originally introduced in~\cite{Carlsson:93,Carlsson:98}. In its simplest form, each agent has access to an independent observation of the underlying fundamental based on which she takes either of the two actions: risky or safe action. The payoff of an agent taking the risky action is monotonically increasing with the number of agents who take the risky action and is a decreasing function of the underlying fundamental. 

Global games have been central to modeling and studying many social coordination phenomena. In \cite{Carlsson:93}, the authors discussed a general form of global games for two players and discussed how the vanishing noise in a public information results in a unique threshold policy. Such analysis and modeling technique was extended and used in \cite{morris1998unique} to model currency attacks. The authors model currency attack using a global game formulation and the authors show that even though the complete information variation of the game has multiple (Nash) equilibria, introducing a small noise into speculators observations about the economic fundamental leads to a unique threshold policy equilibrium. In \cite{morris2004coordination}, a similar idea is used to model the debt crisis and to arrive at a unique threshold policy for the players given the noisy observations of the players. Finally, in \cite{experimental}, the accuracy of the prediction of the global games results have been extensively studied and the authors concluded that: ``Comparing sessions with common and
private information, we observe only small differences in behavior.'' In the engineering domain, the authors of\cite{kanakia2016modeling} have discussed the application of global games in distributed task-allocation in collaborative robotic networks and utilized the same idea to show the uniqueness of threshold policies for task allocation in coordination problems in robotics networks. 

% As an example of such socio-economical behaviors, consider the 

The majority of the past studies on global games have been focused on the information structure where each agent has an independent noisy observation of the underlying economic or social fundamental, i.e.,\ given the fundamental, the  observations of each agent is independent of the rest of the agents' observations. In such settings, it has been shown that, under some condition on the distribution of the underlying fundamental and the independency of agents' observations, there exists an equilibrium strategy with threshold policies based on the private observations of each agent. In \cite{Alireza12}, the case of perfect sharing of information was introduced and the structure of equilibria for global games with perfect sharing of information was discussed. We challenge robustness of these results to sharing of information: we show that even in a simple instance of global games, when agents share information, the intuitive equilibrium doesn't exist. However, we show that with a proper generalization of the notion of \textit{threshold policy}, a symmetric threshold policy (in its extended sense), exists under some conditions on the communication noise. This paper extends the results of the authors in \cite{Touri-global,mahdavifar2015threshold} and provides an extension of the existence of a symmetric equilibria in global games with noisy sharing of information involving more than two players. 

In this paper, global games with noisy sharing of information are introduced where any two agents share their information in a noisy environment. Our contributions are summarized below:
\begin{itemize}
\item
We show the fragility of the existing results in global games by showing that with noisy sharing of information, an intuitive threshold policy which only depends on agents' belief about the underlying fundamental does not lead to a Bayesian Nash equilibrium. 

\item 
We show that if a certain functional equation has a solution, then a symmetric Bayesian Nash equilibrium exists for global games with noisy sharing of information which can be described as a generalized form of the intuitive threshold policies.
\item
We utilize Banach fixed point theorem and show that the functional equation has a solution if the noise and the information structure of the underlying model satisfy certain conditions. This establishes the existence of the generalized threshold strategy equilibria,  when the noise parameters and the information structure satisfy certain conditions. 

\end{itemize}

It is worth noting that this work also relates to many of the recent attempts for understanding the role of information and information structure in multiagent decision making problems \cite{nayyar2011optimal,gupta2014existence,Mahajan,Marden}. In this paper, we consider a static framework, similar to most applications of global games. Global games have been also studied with delay \cite{dasgupta2007coordination} and in dynamic framework \cite{angeletos2007dynamic}, which also relates to networking games and evolution of cooperation in social networks \cite{altman2006survey,ohtsuki2006simple,apicella2012social,tan2015towards,tan2016analysis}. The structure of this paper is as follows. In Section~\ref{sec:setup} we introduce global games with noisy sharing of information, and we mathematically formulate the problem of interest. In Section~\ref{sec:main} we state the three main theorems of the paper on the non-existence of certain intuitive threshold policies, the existence of a general form of threshold policies, and the convergence to this threshold policy via an iterative method. In Section~\ref{proof} the proofs of the three main theorems are established. Some numerical examples are provided in Section\,\ref{numerical}. Finally, we conclude the paper in Section~\ref{sec:conclusion}.

\section{Problem Setup}\label{sec:setup}

In this section, we present the framework of the problem that will be studied in this paper. We study the basic form of the global games with noisy sharing of information. In this setting, we consider a set
$$
[n]:=\{1,\ldots,n\},
$$
 of $n$ agents or players. Each agent has a set of binary actions $A_i=\{0,1\}$. We refer to the action $\alpha_i=1$ as the risky action and $\alpha_i=0$ as the safe action. The payoff of an agent taking the safe action is zero, whereas that of taking the risky action is $\sum_{i=1}^n\alpha_i-\theta$, where $\theta$ is a random variable representing the underlying fundamental in the society. In other words, if $\alpha=(\alpha_1,\ldots,\alpha_n)\in \{0,1\}^n$ is an action profile of the $n$ agents, then the utility of the $i$th agent is the function $u_i$ with
\begin{align}\label{eqn:utility}
u_i(\alpha)=\alpha_i \left(\sum_{j=1}^n\alpha_j-\theta\right).
\end{align}
\noindent
{\bf Remark.} For an intuitive description of the utility functions consider global games in the context of political regime change\cite{edmond}. In this set-up, the parameter $\theta$ represents the power of a political regime that can be overthrown but only if enough citizens participate in an uprising, i.e., take the risky action. It is thus natural to assume that the utility function of agents taking the risky action is increasing in the number of such agents and is decreasing in the fundamental parameter $\theta$. Also, the utility function of agents taking the safe action is considered to be zero. Therefore, utility functions of the form given in \eq{eqn:utility} are natural to use in global game models. The results in the global games literature can often be extended to general utility functions that are monotone in the number of agents taking the risky actions as well as $\theta$.

\textit{Observations and Policies:} In the standard setting for global games, agent $i$ is observing $x_i=\theta+\eta_i$ where $\eta_i$'s are identically and independently distributed random variables \cite{Carlsson:93,Carlsson:98}. In our work, agents share their observations $x_i$ with the other agents in the society through noisy channels. In other words, each agent has its private observation $x_i$ as well as noisy observations of other agents' private observations. Mathematically, we represent agent $i$'s overall observation by a random vector $y_i\in \R^{q_i}$ (which relates to $\theta$ and other agents' observations). The parameter $q_i$ represents the number of observations that agent $i$ has. In this paper, we focus on the case that each agent shares its observation with all other agents. Hence, $q_i = n$, for all $i \in [n]$. We refer to $y_i$ as the (private) information of agent $i$ (about $\theta$). We refer to a measurable function $s_i:\R^{q_i}\to A_i$ that maps a private observation of agent $i$ to one of the two actions as a (pure) \textit{strategy} or \textit{policy}. When $q_i=1$ for some $i$, we say that $s_i$ is a threshold policy if $s_i(y)=1$ for $y\leq t$ and $s_i(y)=0$ for $y>t$, for some threshold value $t$. We denote such a strategy by $s_i(y)=1_{y\leq t}$. In almost all the instances of the global games, the random variables have continuous joint distribution, and hence, the value of the strategy $s_i$ at the threshold value $t$ is practically unimportant.

\textit{Equilibrium:} Our focus in this paper is on the existence of a strategy profile $s=(s_1,\ldots,s_n)$ that results in a Bayesian Nash Equilibrium. To introduce this concept, let $s=(s_1,\ldots,s_n)$ be a strategy profile of the $n$ agents and let $s^{-i}=(s_1,\ldots,s_{i-1},s_{i+1},\ldots,s_n)$ be the strategy profile of the $n-1$ agents except the $i$th agent's strategy. We say that a \textit{best response strategy} to the strategy profile $s^{-i}$ is a strategy $\tilde{s}:\R^{q_i}\to A_i$ such that
\begin{equation}
\tilde{s}(u) = \left\{\hspace{-.05in}
\begin{array}{ll}
0 & \text{if~}1+E[\sum_{j\not=i}s_j(y_j)\mid y_i=u]< E[\theta\mid y_i=u],\\
1 & \text{if~}1+E[\sum_{j\not=i}s_j(y_j)\mid y_i=u]>  E[\theta\mid y_i=u].
\end{array}
\right.
\label{eqn:equilibrium}
\end{equation}
%
%
%
% $\tilde{s}(u)=0$ if 
%\[\] 
%and $\tilde{s}(u)=1$ if 
%\[\]
We denote the set of all best responses to a strategy profile $s^{-i}$ by ${\textit{{\textit{BR}}}}_i(s^{-i})$. 
Finally, we say that a strategy profile $s=(s_1,\ldots,s_n)$ is a \textit{Bayesian Nash Equilibrium} or simply an \textit{equilibrium} if $s_i\in {\textit{BR}}(s^{-i})$ for all $i\in [n]$. 

An extensively studied model in global games is the case where $q_i=1$ for all $i\in [n]$ and $y_i=x_i=\theta+\xi_i$ where $\{\xi_1,\ldots,\xi_n\}$ are {\em independent} and identically distributed $\cN(0,\sigma^2)$ Gaussian random variables for some $\sigma^2>0$. 
Furthermore, it is customary to assume that $\theta$ is picked from a non-informative uniform distribution over $\R$. See \cite{morris2003} and the references therein for further discussion on this assumption. 

In~\cite{morris2003}, it is shown that there exists a symmetric threshold policy on $x_i$'s which corresponds to a Bayesian Nash equilibrium for this instance of global games. In other words, there exists a threshold value $t\in \R$ such that for $x_i\leq t$, agent $i$ chooses to take the risky action and for $x_i>t$, she takes the safe action and such an action profile leads to an equilibrium.  Here, an important fact is that $x_i=E[\theta\mid x_i]$ which means that in such an equilibrium each agent should compare her \textit{expected strength of regime given her private observation} to a threshold and take a proper action accordingly. 

\textit{Our Model for Information Structure:} In this work, we consider global games with noisy sharing of information. We seek to understand the role of information sharing among the agents in the decision making scenarios that are modeled by global games. %For example, what is the structure of an equilibria in bank run where information is being shared between the agents? 

Throughout this work the utility of each agent is given by \eqref{eqn:utility} and the information structure satisfies the following assumption. 
\begin{assumption}\label{assum:information}
We assume that $\theta$ is uniformly distributed over $\R$. Also, agent $i$'s private information is the $n$ dimensional random vector:
\begin{align}
y_i:=(x_i,\{y_{ji} \}_{j \in [n] \setminus \{i\}}),
\label{eqn:y_i}
\end{align}
where $x_{i}=\theta+\xi_i$ for all $i\in [n]$ and $y_{ji}=x_j+\zeta_{ji}$ for $j\not=i$, where $\xi_1,\ldots,\xi_n$ are i.i.d.\ $\cN(0,\sigma^2)$ random variables and $\zeta_{ji}$ are i.i.d.\ $\cN(0,\tau^2)$ random variables that are independent of $\xi_1,\ldots,\xi_n$, and $\sigma^2,\tau^2>0$ are given parameters.
\end{assumption}

The dependence diagram of the random observation vector $y_i=(x_i, \{y_{ji} \}_{j \in [n] \setminus \{i\}})$ that is available to agent $i$ is shown in Figure~\ref{fig:multi-agent}. 

Throughout the rest of the paper, the scalar $\sum_{j \in [n] \setminus \{i\}}y_{ji}$ which is the sum of all the information arrived at agent $i$ from the rest of agents plays a central role. We denote this quantity by $z_i$, i.e., 
\begin{align}\label{eqn:z}
z_i:=\sum_{j \in [n] \setminus \{i\}}y_{ji}.
\end{align} 

We define a global game with noisy sharing of information as below. 
\begin{defn}\label{def:gg}
A game with utility functions $u_1,\ldots,u_n$ described by \eqref{eqn:utility} and information structure satisfying Assumption~\ref{assum:information} is a \textit{\bf global game with noisy information sharing}. 
\end{defn}

The major challenge in analyzing global games with noisy information sharing is the limited information of each agent about the underlying fundamental $\theta$. Note that if all the agents know about the exact realization of $\theta$, and for $\theta\not=n$ either of the symmetric actions $(0,\ldots,0)$ or $(1,\ldots,1)$ would be appealing. However, in the global games, the perfect knowledge of the underlying fundamental is not available to any agent and each agent has a noisy observation of the fundamental $\theta$. 

\begin{figure}
\begin{center}
\includegraphics[width = 0.55\linewidth]{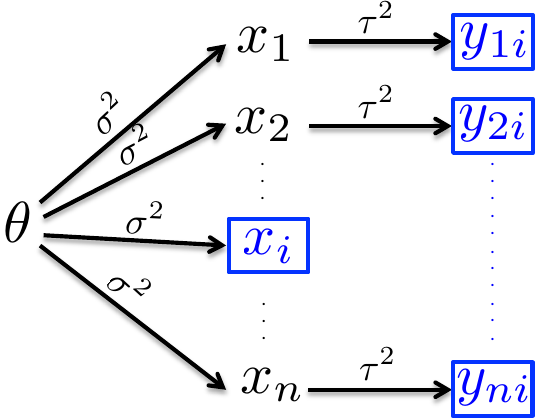}
\end{center}
\caption{The dependence diagram between the observation vector $y_i  = (x_i, \{y_{ji} \}_{j \in [n] \setminus \{i\}})$ of agent $i$. In this diagram, 
$\chi \xrightarrow{\sigma^2} \psi$ implies that $\psi = \chi+ \omega$, where $\chi$ and $\omega$ are independent and $\omega \sim \cN (0,\sigma^2)$.
}
\label{fig:multi-agent}
\end{figure}

\begin{comment}
Let $y_{ji}$ be the information that agent $j$ shares with agent $i$. It is assumed that $y_{ji} = x_j + \zeta_{ji}$, where $\zeta_{ji}$'s are independent $\cN(0,\tau^2)$ random variables. Therefore, agent $i$ will be provided with the vector of observations $(x_i, \{y_{ji} \}_{j \in [n] \setminus \{i\}})$.
\end{comment}

\section{Main Results}\label{sec:main}

In this section, we present the main results of this work including two results. The first result establishes nonexistence of an \textit{intuitive} equilibrium for global games with noisy sharing of information. The second result shows that under a certain regime, there exists a threshold policy based on each agent's private information and the (noisy) information that is shared with her by the other agents. 

\subsection{Nonexistence of an Intuitive Equilibrium}
In many instances of global games, one can show that there exists a threshold policy based on the expected value of the underlying fundamental given each agent's information. More precisely, let 
\[\bar{\theta}_i(y_i)=E[\theta\mid y_i].\]
It has been shown that if there is no sharing of information (i.e., $y_i=x_i$) and $\theta$ is uniformly distributed over $\R$ or $\theta$ has a Gaussian prior, then there exists a threshold value $t$ such that the policy $(1_{\bar{\theta}_1(y_1)\leq t},\ldots, 1_{\bar{\theta}_n(y_n)\leq t})$ is a Bayesian Nash Equilibrium for the global games described above \cite{morris2000}. 

In the case of global games with noisy sharing of information, one may hope to have a similar result, i.e., there exists a threshold value $t$ such that if everyone compares her expected value of the fundamental variable $\theta$ given her own information, the policy $(1_{\bar{\theta}_1(y_1)\leq t},\ldots, 1_{\bar{\theta}_n(y_n)\leq t})$ would be a Bayesian Nash Equilibrium. For example, in the case of a bank run, one may speculate that if each agent compares her expected strength of the economy given her own information, she should decide whether to take her money out of the bank or not, and this behavior would result in an equilibrium. However, the following result shows that such an \textit{intuitive} equilibrium for global games with noisy sharing of information does not exist. 

\begin{theorem}\label{thrm:nonexistence}
Consider the global game with noisy sharing of information as described in Definition \ref{def:gg}. Then, there do not exist threshold values $t_1,\ldots,t_n$ such that the policy $(1_{\bar{\theta}_1(y_1)\leq t_1},\ldots, 1_{\bar{\theta}_n(y_n)\leq t_n})$ is an equilibrium. 
\end{theorem}

\subsection{Existence of a Threshold Policy Equilibrium}
In light of Theorem~\ref{thrm:nonexistence}, one may pose the question as how one can extend the existence of threshold policies for conditionally independent signals to the case of interdependent private signals. Here, we show that indeed we can extend such an existence result to the case of global games with sharing of information. 

Before presenting this result, the notion of threshold policies is extended from the case of a scalar private information to multi-dimension information vector. Consider the information available to agent $i$, i.e., $y_i=(x_i,\{y_{ji} \}_{j \in [n] \setminus \{i\}})$. For a function $h:\R^{n}\to\R$, we will be focusing on the symmetric threshold strategies $1_{h(y_i) \geq 0}$. A strategy is called symmetric if it does not depend on the index of the agents. Namely, in this case, the function $h$ is the same for all the agents. We will further limit our attention to the class of functions $h:\R^{n}\to\R$ which are continuous and strictly decreasing with respect to any of the $n$ input parameters, while the other input parameters are fixed. Monotone strategies are natural candidates to consider because if taking risky action is not appealing for a given observation, it is natural to consider strategies that assign the safe action to any larger observation. More precisely, let $y'_i$ be smaller than $y_i$ component-wise. If agent $i$ takes a risky action by observing $y_i$, it is intuitively expected that it takes the risky action by observing $y'_i$ as well. In other words, assuming a threshold strategy $1_{h(y_i) \geq 0}$, we expect that if $h(y_i) \geq 0$, then $h(y'_i) \geq 0$. A symmetry condition on $h$ is also required which will be clarified later in this section.

With this, the main existence results of this paper are stated in the next two theorems. 
\begin{theorem}\label{thrm:thm2}
Let $h: \R^n \rightarrow \R$ be a continuous and component-wise strictly decreasing function. Then $h$ leads to a threshold policy equilibrium if it is the solution to a certain functional equation characterized by the noise variance parameters of the system. 
\end{theorem}

\begin{theorem}\label{thrm:thm3}
For any given $n$, there exists an unbounded set of parameters $\Theta\subset \R^{+2}$ such that for any $(\sigma,\tau)\in \Theta$ there exists a symmetric Bayesian Nash Equilibrium $(1_{h(y_1) \geq 0},\ldots,1_{h(y_n) \geq 0})$ for some continuous and component-wise decreasing function $h:\R^n\to\R$. Furthermore, $h$ can be approximated with an arbitrary precision (in $\cL_\infty$ norm).
\end{theorem}

The proof of \Tref{thrm:thm3} will be based on defining a contraction mapping on the space of feasible strategy functions $h$ and then utilizing Banach fixed point theorem to establish the existence of $h$ that solves the functional equation of \Tref{thrm:thm2}. Consequently, it is shown that an iterative method of applying the contraction mapping will result in the convergence to the desired $h$ which can be used as an approximation method to find $h$ with arbitrary precision.

We will also characterize a sufficient condition on the noise parameters $(\sigma,\tau)$ to be contained in the set $\Theta$ of \Tref{thrm:thm3}. In particular, it is shown that for a fixed ratio of $r=\frac{\tau}{\sigma}$, $(\sigma,r \sigma) \in \Theta$, for large enough $\sigma$. This implies that the set $\Theta$ is unbounded in any direction and the set $\Theta^c$ is bounded in any dierction. Roughly speaking, for a sufficiently noisy environment there exists a threshold type Bayesian Nash Equilibrium. 
 
One may conjecture that \Tref{thrm:thm3} should hold for the whole set of noise parameters and it should be independent from the noise parameters, i.e., $\Theta = \R^{+2}$. However, we indeed conjecture that such a symmetric equilibrium does not exist for arbitrary $\sigma,\tau>0$. 

The above three theorems suggest that in the settings that global games are relevant, one can not be oblivious to the information structure details and only relies on the estimate of the underlying fundamental $\theta$.

\section{Proofs}
\label{proof}

In this section, we provide the proof of the main results discussed in Section~\ref{sec:main}. 

The following calculations for the statistics of conditional Gaussian random vectors are needed in the proof of the main results. Next lemma derives the probability distribution of $\theta$ conditioned on the observation vector $y_i$ of the $i$-th agent.

\begin{lemma}\label{lemma:properties}
	Suppose that Assumption~\ref{assum:information} holds.  Let $\eta^2_n$ be a scalar such that
	$$
	\frac{1}{\eta^2_{n}} = {\frac{1}{\sigma^2} + \frac{n-1}{\sigma^2 + \tau^2}}.
	$$
	Further, let $a_n :=    \frac{\eta_n^2}{\sigma^2}$ and $b_n :=    \frac{\eta_n^2}{\sigma^2 + \tau^2}$.
	Then, conditioned on agent $i$'s observation $y_i$, $\theta$ is given by
	%$\ot(x,y)=\frac{(\sigma^2+\tau^2)x+\sigma^2y}{2\sigma^2+\tau^2}$. Then,
	\begin{align*}
	\theta  = a_n x_i + b_nz_i + \ep_\theta,
	\end{align*}
	where $\ep_\theta$ is a $\cN(0,\eta^2_n)$ Gaussian random variable independent of $y_i$.
\end{lemma}
The proof can be found in Appendix.

To investigate the existence of a threshold policy for agent $i$, we further need to derive the distribution of agent $k$'s observation vector $y_k$ given agent $i$'s observation $y_i$. In a sense, this is agent $i$'s perception of what is available to agent $k$.

\begin{lemma}
	\label{itm:dist}
	Let Assumption~\ref{assum:information} hold. Further, let
	$$
	\frac{1}{\gamma_n^2} =   \frac{1}{\tau^2} + \frac{1}{\eta_{n-1}^2 + \sigma^2},
	$$
	and define
	$c_n :=   \frac{\gamma_n^2}{\eta_{n-1}^2 + \sigma^2}$ and $d_n := \frac{\gamma_n^2}{\tau^2}$.
	Then, conditioned on $y_i$, we can write
	$y_k=(x_k, \{y_{lk} \}_{l \in [n] \setminus \{k\}} )$  as jointly Gaussian random variables defined by:
	$$
	x_k =\overline{x_k} + \ep_k ,~~
	y_{ik} = x_i + \ep_{ik},~~
	y_{lk} = \overline{y_{lk}}+ \ep_{lk},
	$$
	where $l \neq i$ and\footnote{Note that $\overline{x_k}$ and $\overline{y_{lk}}$ are functions of the observation vector $y_i$ of agent $i$ but the dependence is left implicit for brevity.}
	\begin{align}
	\left\{
	\begin{array}{l}
	\overline{x_k} = c_na_{n-1}  x_i + d_n y_{ki} + c_n b_{n-1}\sum_{j \in [n] \setminus \{i,k\}} y_{ji},\\
	\overline{y_{lk}} = c_na_{n-1}  x_i + d_n y_{li} + c_n b_{n-1}\sum_{j \in [n] \setminus \{i,l\}} y_{ji},
	\end{array}
	\right.
	\label{eq:mean}
	\end{align}
	and $(\ep_k, \{\ep_{lk}\}_{l \in [n] \setminus \{k\}  } )$ are jointly zero mean Gaussian random variables independent of $(x_i, \{y_{ji} \}_{j \in [n] \setminus \{i\}} )$; and $\ep_{ik}$ is $\cN(0, \tau^2)$ independent of $\epsilon:=(\ep_k, \{\ep_{lk}\}_{l \in [n] \setminus \{i,k\}  } )$.
%Finally, the covariance of $\epsilon$ is  characterized by
%\begin{align*}
%E[\ep_k^2]&= \gamma^2_n,\\
%E[\ep_k \ep_{lk}] &=,\\ 
%E[\ep_{lk} \ep_{mk}] &=,\\
%E[\ep_{lk}^2] & = \gamma^2_n+ \tau^2,
%\end{align*}
%for $l,m \in [n] \setminus \{i,k\} $ and $l \neq m$.
\end{lemma}
The proof can be found in Appendix.

\subsection{Proof of Theorem \ref{thrm:nonexistence}}

The underlying idea of the proof can be summarized as follows. Assume to the contrary that $(1_{\bar{\theta}_1(y_1)\leq t_1},\ldots, 1_{\bar{\theta}_n(y_n)\leq t_n})$ is an equilibrium for some $t_1,\ldots,t_n\in \R$, where $\bar{\theta}_i(y_i) = a_nx_i+b_nz_i$ by  \Lref{lemma:properties}. Then the best response strategy of agent $i$, given the strategies of other agents $1_{a_nx_j+b_nz_j\leq t_1}$ is given by \eq{eqn:equilibrium}. This must be equivalent to $1_{a_nx_i+b_nz_i\leq t_1}$ by the definition of the Bayesian Nash equlibrium. However, we use the non-singularity of the system of linear equations describing the noise variances to show that this can not happen resulting in a contradiction.

   Let $\beta_n, \gamma_n, \delta_n$ be defined as follows:
   \begin{align*}
   \beta_n &:= b_n + a_n c_n a_{n-1} + (n-2) b_n c_n a_{n-1}\\
   \gamma_n &:= a_nd_n + (n-2) b_n c_n b_{n-1}\\
 \delta_n &:= a_nc_nb_{n-1} + b_n d_n + (n-3) b_n c_n b_{n-1}
   \end{align*}
   Then let $V = [v_{ij}]_{n \times n}$ be a matrix with the following entries:
   \begin{align*}
   v_{ij} = \left\{
	\begin{array}{l}
	a_n,\ \text{if}\ i = 1, j=1,\\
	b_n,\ \text{if}\ i = 1, 1 < j \leq n ,\\
	\beta_n, \ \text{if}\ 1 < i \leq n , j = 1,\\
	\gamma_n,\ \text{if}\ 1 <i = j\leq n, \\
	\delta_n,\ \text{if}\ 1 < i\neq j \leq n, \\
	\end{array}
	\right..
   \end{align*}
   In other words, $V$ has the following structure:
   \begin{align*}
   V = \left(
   \begin{array}{ccccc}
    a_n    & b_n & b_n&\ldots& b_n\\
   \beta_n &\gamma_n&\delta_n&\ldots&\delta_n\\
   \beta_n &\delta_n &\gamma_n&\ldots&\delta_n\\
	\vdots &\vdots &\vdots&\ddots&\vdots\\
	\beta_n &\delta_n &\delta_n&\ldots&\gamma_n
   \end{array}\right).
   \end{align*}
   
   By Lemma~\ref{lemma:properties} and \Lref{itm:dist}, we have 
   \be{main-thm11}
   \bigl(\bar{\theta}_i(y_i), \left\{\EXP{\bar{\theta}_j(y_j)\mid y_i}\right\}_{j \neq i} \bigr) = V y_i ^T
   \ee
 It is assumed that $\tau > 0$, i.e., information sharing between agents is noisy and not perfect. This leads to $a_n > b_n$ and consequently $\gamma_n > \delta_n$. We next show that this implies that $V$ is non-singular. By subtracting the first column of $V$ scaled by $\delta_n/\beta_n$ from all other columns we get 
%$V' = [v'_{ij}]_{n \times n}$ with
%   \begin{align*}
% v'_{ij} = \left\{
%	\begin{array}{l}
%	a_n,\ \text{if}\ i = 1, j=1,\\
%	b_n-\delta_n a_n / \beta_n,\ \text{if}\ i = 1, 1 < j \leq n ,\\
%	\beta_n, \ \text{if}\ 1 < i \leq n , j = 1,\\
%	\gamma_n - \delta_n,\ \text{if}\ 1 <i = j\leq n, \\
%	0,\ \text{if}\ 1 < i\neq j \leq n, \\
%	\end{array}
%	\right.
%   \end{align*}
   \begin{align*}
   V' = \left(
   \begin{array}{ccccc}
   a_n    & b_n-\frac{\delta_n}{\beta_n}a_n& b_n-\frac{\delta_n}{\beta_n}a_n&\ldots& b_n-\frac{\delta_n}{\beta_n}a_n\\
   \beta_n &\gamma_n-\delta_n&0&\ldots&0\\
   \beta_n &0 &\gamma_n-\delta_n&\ldots&0\\
   \vdots &\vdots &\vdots&\ddots&\vdots\\
   \beta_n &0 &0&\ldots&\gamma_n-\delta_n
   \end{array}\right).
   \end{align*}

It can be observed that $b_n-\delta_n a_n / \beta_n < 0$ and hence, by subtracting the first row scaled by $\frac{\bet_n}{a_n}$  from the other rows, $V'$ is made lower triangular with strictly positive diagonal elements. Therefore, $V'$ and consequently $V$ are non-singular.

   Using the fact that $V$ is non-singular together with \eq{main-thm11} one can find $y_i \in \R^n$ such that  $\bar{\theta}_i(y_i) = a_nx_i+b_nz_i = t_i - \epsilon$, for an arbitrary $\epsilon >0$, and $\EXP{\bar{\theta}_j(y_j)\mid y_i} > M$, for $j \neq i$, for any arbitrarily large number $M \in \R^+$. Note that the variance of $\bar{\theta}_j(y_j)$ given $y_i$ is a constant function of $\sigma^2$ and $\tau^2$ (by Lemma~\ref{itm:dist}). Therefore, using Chebyshev's inequality, for the given threshold values $t_1,\ldots,t_n$, one can find $y_i\in \R^n$ such that $\bar{\theta}_i(y_i) = t_i - \epsilon$ and
   \begin{align*}
   P(\bar{\theta}_j(y_j)\leq t_j \mid y_i)\leq \epsilon.
   \end{align*}
   But by the structure of an equilibrium \eqref{eqn:equilibrium}, we should have
   \begin{align*}
   (n-1)\epsilon+1&\geq \sum_{j\not=i}P(\bar{\theta}_j(y_j)\leq t_j \mid y_i)+1\cr
   &\geq \bar{\theta}_i(y_i) = t_i - \epsilon.
   \end{align*}
%   But this inequality holds for any $\epsilon>0$ and hence, $t_1\leq 1$. Note that for any $t_1,t$, the set
%   \begin{align*}
%   E_{t_1,t}&=\{(x,y)\mid \bth(x,y)\leq t_1,\cr
%   &\EXP{\bth(x_2,y_2)\mid x_1=x,y_1=y}-\bth(x,y)\geq t\},
%   \end{align*}
%   has a non-zero Lebesgue measure and hence, $P((x_1,y_1) \in E_{t_1,t})>0$.
 Since the above inequality must hold for any $\epsilon > 0$, we have $t_i \leq 1$.

   On the other hand, using the same argument for any $\epsilon>0$, one can find $y_i$ such that $\bar{\theta}_i(y_i)=t_i+\epsilon$ and $\EXP{\bar{\theta}_j(y_j)\mid y_i} < -M$, for $j \neq i$, for any arbitrarily large number $M \in \R^+$. Hence,
   \begin{align*}
   t_i+\epsilon&\geq \bar{\theta}_i(y_i)\cr
   &> E(\sum_{j\not=i}1_{\bar{\theta}_j(y_j)\leq t_j}\mid y_i)+1\cr
   &\geq 1+(1-\epsilon)(n-1)=n-(n-1)\epsilon.
   \end{align*}
   Since, this holds for any $\epsilon>0$, it follows that $t_i\geq n$ which contradicts $t_i\leq 1$. Therefore, such a threshold equilibrium does not exist.

The above proof can be extended to a more general case of utility functions that are monotone with respect to the actions of the players and also are continuous functions of the underlying fundamental $\theta$.

\subsection{Proof of Theorem \ref{thrm:thm2}}
\label{sec:proof2}
We start by introducing the set of threshold policies that we will be focusing on. Such policies are characterized by certain \emph{threshold functions}. We impose some natural constraints on the threshold function $h$. The description of the best response strategies, formulated in \eq{eqn:equilibrium}, provides a necessary condition on $h$ that leads to threshold-type Bayesian Nash equilibrium. This condition can be interpreted as $h$ being a fixed point to a certain functional equation resulting from \eq{eqn:equilibrium}. Then we exploit natural properties of $h$, such as being continuous, monotone ,and symmetric, to conclude that such condition is also sufficient for having a threshold-type Bayesian Nash equilibrium characterized by $h$.

Consider the information available to agent $i$, i.e., $y_i=(x_i,\{y_{ji} \}_{j \in [n] \setminus \{i\}})$. For a function $h:\R^{n}\to\R$, we will be focusing on the threshold strategies $1_{h(y_i) \geq 0}$. We will further limit our attention to the class of functions $h:\R^{n}\to\R$ which are continuous and strictly decreasing with respect to any of the $n$ input parameters, while the other input parameters are fixed. A symmetry condition on $h$ is also required which will be clarified later in this section.

 The goal is to show there exists, under some conditions on noise parameters $\sigma^2,\tau^2$, a function $h: \R^n \rightarrow \R$ with the above conditions, such that the strategy profile $s=(1_{h(y_1) \geq 0},\ldots,1_{h(y_n)\geq 0})$ is an equilibrium for the global games with noisy sharing of information, where $s_i=1_{h(y_i) \geq 0}$ is the threshold policy that prescribes:
\be{eq-policy}
\begin{split} 
	\left\{
   \begin{array}{l}
\alpha_i = 1\ \ \text{if}\ \ h(y_i) \geq 0 \\
\alpha_i = 0\ \ \text{if}\ \ h(y_i) < 0
\end{array}
\right.
\end{split}
\ee
for agent $i$. We refer to such a function $h$ as a \emph{threshold function}, and we refer to the resulting symmetric strategy profile 
\begin{align}\label{eqn:threshold}
s=(1_{h(y_1) \geq 0},\ldots,1_{h(y_n)\geq 0})
\end{align}
as a symmetric threshold policy. If the strategy profile given in \eqref{eqn:threshold} is an equilibrium for the underlying game, we say that the threshold function $h$ leads to a (symmetric) threshold policy equilibrium.

By the definition of a (Bayesian Nash) equilibrium, the function $h$ leads to a threshold policy equilibrium if a best response of any agent $i \in [n]$ is characterized by the threshold policy described in \eq{eq-policy}. In other words, we have
\be{eq1}
\begin{split}
1 + \sum_{k \in [n] \setminus \left\{i\right\}} P\Bigl( h(y_k) \geq 0 \mid y_i \Bigr)&\geq 
E( \theta \mid y_i) \\
& = a_{n} x_i + b_{n} z_i
\end{split}
\ee
if and only if $h(y_i) \geq 0$, for any $i \in [n]$, where $z_i$ is given in \eq{eqn:z}. 

\noindent
{\bf Remark.}  
It is shown in Theorem~\ref{thrm:nonexistence} that an equilibrium with a threshold policy on $E( \theta | y_i)$ does not exist. In light of the definition of threshold functions provided here, this result can be stated as the threshold function $h(x,y) = a_n x + b_n y + c$, where $c \in \R$ is a constant, does not lead to a threshold policy equilibrium. We emphasize that our definition of a threshold function only considers a symmetric threshold policy (if it exists), where all agents take actions according to the same threshold function $h$. 

Consider a continuous threshold function $h$ that leads to a Bayesian Nash equilibrium. It can be observed that if $h(y_i) = 0$, then \eq{eq1} also turns into equality. In other words, we have
\be{eq-fxp}
\begin{split}
a_{n} x_i + b_{n} z_i= 1 + \sum_{k \in [n] \setminus \left\{i\right\}} P\Bigl( h(y_k) \geq 0 | y_i\Bigr).
\end{split}
\ee
 Note that given the parameters of the system, the right hand side of \eq{eq-fxp} can be regarded as an operation on the function $h$. The definition of this operation will be provided later in this section. This motivates us to define a \emph{fixed point threshold} function as follows: 
 
\begin{definition}
\label{def0}
We say that $h: \R^n \rightarrow \R$ is a fixed point threshold function if for any choice of $y_i = (x_i, \{y_{ji} \}_{j \in [n] \setminus \{i\}})$, we have $h(y_i) = 0$ if and only if \eq{eq-fxp} also holds.
\end{definition}

We note here that a fixed point threshold function $h(\cdot)$ is not necessarily unique. In fact, the set of transition points $\{r\mid h(r)=0\}$ plays a central role in defining the strategy profile as in \eq{eqn:threshold}, rather than $h$ itself. Any other strictly decreasing function $g$ that has the same set of transition points defines the same strategy profile as $h$ does. 

For notational convenience, we let $$y_{i\setminus k}:=(x_i,\{y_{ji} \}_{j \in [n] \setminus \{i,k\}})$$ for $i,k \in [n]$. Note that if $h: \R^n \rightarrow \R$ is a continuous fixed point threshold function, then for any $k \in [n]\setminus\left\{i\right\}$ and $y_{i \setminus k}$, there exists $y_{ki} \in \R$ such that $h(y_i) = 0$, where $y_i$ is defined in~\eqref{eqn:y_i}. The reason is that the right hand side of \eq{eq-fxp} is bounded between $1$ and $n$, while the left hand side of \eq{eq-fxp} is unbounded in terms of $y_{ki}  \in \R$, while $(x_i, \{y_{ji} \}_{j \in [n] \setminus \{i,k\}})$ is fixed. Therefore, there exists $y_{ki}  \in \R$ such that  \eq{eq-fxp} turns into equality and since $h$ is a fixed point threshold function, by definition, $h(y_i) = 0$. We let $\cI h  (x_i, \{y_{ji} \}_{j \in [n] \setminus \{i,k\}})$ to denote the solution $y_{ki}$ for $h(y_i) = 0$. Furthermore, if $h$ is strictly decreasing with respect to any of its input parameters, $\cI h  (x_i, \{y_{ji} \}_{j \in [n] \setminus \{i,k\}})$ is unique. In fact, $\cI h: \R^{n-1} \rightarrow \R$ is also a strictly decreasing function. 

As mentioned before, the function $h$ is not unique, however, $\cI h$ is unique under some additional conditions. Therefore, instead of explicit characterization of $h$, we will be solving the fixed point equation for $\cI h$ and then, we pick $h$ as follows:
\be{eq-h}
h(y_i) := \cI h(y_{i \setminus k}) - y_{ki}
\ee

As discussed earlier, a symmetry condition on $h$ with respect to $\{y_{ji} \}_{j \in [n] \setminus \{i\}}$ is naturally needed in order to have a symmetric threshold policy equilibrium. In fact all the other agents look the same to the agent $i$ and the observations $\{y_{ji} \}_{j \in [n] \setminus \{i\}}$ follow the same model, as illustrated in Figure\,\ref{fig:multi-agent}, hence the indexing of other agents should not matter. The symmetry condition can be that $h(y_i)$ is the same if $\{y_{ji} \}_{j \in [n] \setminus \{i\}}$ is permuted. However, this is too general for our purpose and $h$ as constructed in \eq{eq-h} may not satisfy this condition. In fact the symmetry condition only matters when $h(y) = 0$, because the threshold policy will be uniquely determined given the set of solutions for $h(y) = 0$ when $h$ is strictly decreasing. The symmetry condition is then defined as follows:
\begin{definition}
\label{def-sym}
We say that $h: \R^n \rightarrow \R$ is a symmetric threshold function if for any root $y = (x_i, \{y_{ji} \}_{j \in [n] \setminus \{i\}})$ of $h$, $y$ with a permuted $\{y_{ji} \}_{j \in [n] \setminus \{i\}}$ is also a root of $h$.
\end{definition}

Note that if $h$ is strictly decreasing and symmetric fixed point function, the function $\cI h: \R^{n-1} \rightarrow \R$ does not depend on the choice of index $k$. 

The following theorem is the main result of this section and presents a sufficient condition on $h$ to be a threshold function. 
\begin{theorem}
\label{thm-main}
Let $h: \R^n \rightarrow \R$ be strictly decreasing and symmetric fixed point threshold function according to Definition\,\ref{def0}. Then $h$ leads to a threshold policy equilibrium. 
\end{theorem} 

\begin{proof}
For $i \in [n]$, let $y_i=(x_i, \{y_{ji} \}_{j \in [n] \setminus \{i\}})$ denote the observations of agent $i$. We fix $i$ and consider two different cases:

\noindent
Case 1: $h(y_i) \geq 0$.  
\\Note that $\cI h : \R^{n-1} \rightarrow \R$ is a strictly decreasing function and $h(y_k) \geq 0$ if and only if $y_{ik} \leq \cI h  (x_k, \{y_{lk} \}_{l \in [n] \setminus \{i,k\}})$. Using \Lref{itm:dist} we have
\begin{equation}
\begin{split}
\label{lem1-1}
&P\Bigl( h(y_k) \geq 0 \mid y_i\Bigr) \\ 
& = P\Bigl( y_{ik} \leq \cI h  (y_{k \setminus i})  \mid y_i\Bigr),\\
& = P\Bigl(  x_i + \eps_{ik} \leq \cI h(\overline{x_k}+\ep_k,\{\overline{y_{lk}}+\ep_{lk}\}_{l \in [n] \setminus \left\{i,k\right\}}) \Bigr),
\end{split}
\end{equation}
where $\overline{x_k}$ and $\overline{y_{lk}}$ are the means of $x_k$ and $y_{lk}$ conditioned on $y_i$, respectively, and are derived in~\eqref{eq:mean}.
Let $\beps = \bigl(\ep_k, \{\ep_{lk}\}_{l \in [n] \setminus \{i,k\}  } \bigr)$ and $f_{\beps} (\beps)$ denote the joint probability density function (PDF) of the Gaussian random variables $(\ep_k, \{\ep_{lk}\}_{l \in [n] \setminus \{i,k\} })$. By \Lref{itm:dist}, $\eps_{ik}$ is independent of $\beps$ and thus, \eq{lem1-1} can be rewritten as in~\eqref{lem1-2}, shown on top of the next page.
\begin{figure*}[t]
\begin{center}
\line(1,0){400}
\end{center}
\begin{equation}
\label{lem1-2}
P\Bigl( h(y_k) \geq 0 \mid y_i\Bigr)  
= \int_{\R^{n-1}}\hspace{-.05in} \varphi \bigl(\frac{1}{\tau} (\cI h(\overline{x_k}+\ep_k,\{\overline{y_{lk}}+\ep_{lk}\}_{l \in [n] \setminus \left\{i,k\right\}}) - x_i)  \bigr)
f_{\beps}(\beps) d\beps
\end{equation}
\begin{center}
\line(1,0){400}
\end{center}
\end{figure*}
%\begin{equation}
%\begin{split}
%\label{lem1-2}
%&P\Bigl( h(y_k) \geq 0 \mid y_i\Bigr)  \\
%&= \int_{\R^{n-1}}\hspace{-.05in} \varphi \bigl(\frac{1}{\tau} (\cI h(\overline{x_k}+\ep_k,\{\overline{y_{lk}}+\ep_{lk}\}_{l \in [n] \setminus \left\{i,k\right\}}) - x_i)  \bigr)\\
%& \hspace{13mm} f_{\beps}(\beps) d\beps
%\end{split}
%\end{equation}

In~\eqref{lem1-2}, $\varphi$ is the cumulative distribution function (CDF) of the normal distribution with unit variance. For $j \in [n]\setminus \left\{i\right\}$, let $y'_{ji} \geq y_{ji} $ such that $h(y'_i) = 0$, where $y'_i=(x_i, \{y'_{ji} \}_{j \in [n] \setminus \{i\}})$. Let also $\overline{x_k}'$ and $\overline{y_{lk}}'$ be defined with respect to $x_i$ and $y'_{ji}$ as in \Lref{itm:dist}. Then by following the same arguments and by noting that $y_{ik} = x_i + \eps_{ik}$ does not change while changing $y_{ji}$ to $y'_{ji}$, ~\eqref{lem1-3} follows, 
\begin{figure*}[t]
\begin{center}
\line(1,0){400}
\end{center}
\begin{equation}
\label{lem1-3}
P\Bigl( h(y_k) \geq 0 \mid y'_i \Bigr) 
= \int_{\R^{n-1}}\hspace{-.05in} \varphi \bigl(\frac{1}{\tau} (\cI h(\overline{x'_k}+\ep_k,\{\overline{y'_{lk}}+\ep_{lk}\}_{l \in [n] \setminus \left\{i,k\right\}}) - x_i)  \bigr)
 f_{\beps}(\beps) d\beps,
\end{equation}
\begin{center}
\line(1,0){400}
\end{center}
\end{figure*}
%\begin{equation}
%\begin{split}
%\label{lem1-3}
%&P\Bigl( h(y_k) \geq 0 \mid y'_i \Bigr) \\
%&= \int_{\R^{n-1}}\hspace{-.05in} \varphi \bigl(\frac{1}{\tau} (\cI h(\overline{x'_k}+\ep_k,\{\overline{y'_{lk}}+\ep_{lk}\}_{l \in [n] \setminus \left\{i,k\right\}}) - x_i)  \bigr)\\
%& \hspace{13mm} f_{\beps}(\beps) d\beps,
%\end{split}
%\end{equation}
where $y'_i=(x_i,\{y'_{ji}\}_{j \in [n] \setminus \{i\}})$. Observe that $\overline{x_k} \leq \overline{x_k}'$ and $\overline{y_{lk}} \leq \overline{y_{lk}}'$. Also, note that $\cI h$ is a decreasing function and $\varphi$ is an increasing function. Therefore, \eq{lem1-2} together with \eq{lem1-3} imply that
\be{lem1-4}
\begin{split}
&P\Bigl( h(y_k) \geq 0 \mid y_i\Bigr)  \geq P\Bigl( h(y_k) \geq 0 \mid y'_i\Bigr).
\end{split}
\ee
By summing \eq{lem1-4} over all $k \in [n]\setminus \left\{i\right\}$ and using the fixed point property of $h$ at $(x_i, \{y'_{ji} \}_{j \in [n] \setminus \{i\}})$ as described in Definition\,\ref{def0}, we get
\be{lem1-5}
\begin{split}
& 1 + \sum_{k \in [n] \setminus \left\{i\right\}} P\Bigl( h(y_k) \geq 0 \mid y_i\Bigr) \cr
&\qquad\geq 1 + \sum_{k \in [n] \setminus \left\{i\right\}} P\Bigl( h(y_k) \geq 0 \mid y'_i\Bigr) \\
& \qquad= a_{n} x_i + b_{n} \sum_{j \in [n] \setminus \left\{i\right\}} y'_{ji} \geq a_{n} x_i + b_{n}z_i\\
& \qquad= E( \theta \mid y_i).
\end{split}
\ee
Therefore, the best response of agent $i$ is to take the risky action.
\\
\noindent
Case 2: $h(y_i) < 0$. 
\\In this case we take $y'_{ji} < y_{ji} $ such that $h(y'_i) = 0$. The inequalities in \eq{lem1-4} and \eq{lem1-5} will be reversed and the best response of agent $i$ is to take the safe action. This will complete the proof.
\end{proof}

\subsection{Proof of \Tref{thrm:thm3}}
\label{convergence}
In this section, we analyze the convergence of an iterative scheme for finding the threshold function and consequently the threshold policy equilibrium. To this end, a certain operator $\cT$ is defined that captures the update of a threshold policy when agents update their strategy according to the best response strategy rule.   Sufficient conditions are derived to guarantee that the operator $\cT$ becomes a contraction mapping. Then the Banach fixed point theorem is utilized to show the convergence of an iterative scheme, that applies $\cT$ iteratively to an initial function, to a fixed-point threshold function. Then \Tref{thrm:thm2} is utilized to show that such threshold function results in a Bayesian Nash equilibrium. This complete the proof of \Tref{thrm:thm3}. 

As the first step, we find $\cI h$ such that the threshold function $h$, as derived in \eq{eq-h}, satisfies the conditions of \Tref{thm-main}. By replacing $y_{ki}$ with $\cI h(y_{i \setminus k})$ in \eq{eq-fxp} and using derivation of $P\Bigl( h(y_k) \geq 0 \mid y_i\Bigr)$ in \eq{lem1-2}, \eq{eq-fxp} can be turned into a fixed point equation for the function $\cI h$. Let
\be{def-g}
g(y_{i\setminus k}) := a_n x_i + b_n \cI h(y_{i\setminus k})+b_n \sum_{j \in [n] \setminus \{k,i\}} y_{ji}.
\ee
Note that $g(y_{i\setminus k})$ is simply the left-hand side of \eq{eq-fxp}. We finally arrive at the definition for the fixed point function $g: \R^{n-1} \rightarrow \R$ as follows.  

\begin{definition}
\label{def1}
We call $g: \R^{n-1} \rightarrow \R$ with the input $y \in \R^{n-1}$ to be a fixed point function if 
\be{eq4}
\begin{split}
g(y)= 1+\sum_{l \in [n]\setminus \{i\}} \int_{\R^{n-1}} \varphi\left(\cM_{\beps,l} g(y)\right) f_{\beps} (\beps) d\beps,
\end{split}
\ee
where $f_{\beps} (\beps)$ denote the joint probability density function (PDF) of the Gaussian random variables $(\ep_l, \{\ep_{jl}\}_{j \in [n] \setminus \{i,l\} })$ and $\cM_{\beps,l} g: \R^{n-1}\to\R$ is defined as
\be{defM}
\begin{split}
\cM_{\beps,l} g(y) &= \frac{1}{b_n\tau} \bigl(g(\overline{x_l}+\ep_l,\{\overline{y_{jl}}+\ep_{jl}\}_{j \in [n] \setminus \left\{i,l\right\}})\\
&- a_n (\overline{x_l}+\eps_l)- b_n \sum_{j \in [n] \setminus \{i,l\}} (\overline{y_{jl}}+\ep_{jl}) - b_n x_i\bigr),
\end{split}
\ee
where $\overline{x_l}$ and $\overline{y_{jl}}$ are the means of $x_l$ and $y_{jl}$ derived in~\eqref{eq:mean} in terms of $y_{i\setminus k} = y$ and
$$
y_{ki} = \frac{1}{b_n}g(y_{i\setminus k}) - \frac{a_n}{b_n}x_i - \sum_{j \in [n] \setminus \{k,i\}} y_{ji}.
$$
\end{definition}

\begin{definition}
\label{def-T}
 We define the operator $\cT$ to be the operator that maps a sufficiently well-behaved function $g:\R^{n-1}\to\R$ to $\gd:\R^{n-1}\to\R$ defined by
\be{eq5}
\begin{split}
\gd(y) :=& 1+\sum_{l \in [n]\setminus \{i\}} \int_{\R^{n-1}} \varphi(\cM_{\beps,l} g(y)) f_{\beps} (\beps) d\beps,
\end{split}
\ee
where $\cM_{\beps,l} $ is defined in \eq{defM}. 
\end{definition}

If $g$ is a fixed point function, according to Definition\,\ref{def1}, then $\gd  = g$, by definition of $\gd$ in \eq{eq5}. Therefore, finding a fixed point function $g$ is equivalent to finding a fixed point for the operator $\cT$. 

In the subsequent discussion, we will derive conditions that will characterize the term \textit{sufficiently well-behaved} in the above statement. First we notice that $\cT$ can be viewed as a mapping that maps the space of measurable functions $g: \R^{n-1}\rightarrow [1,n]$ to itself. This follows from the fact that $\varphi(\alpha)\in [0,1]$ for all $\alpha\in \R$. 

To find a fixed point for the operator $\cT$, the structure of the fixed point equation \eqref{eq5} suggests the investigation of the iteration
\begin{align}\label{eqn:iter}
g^{(t+1)}=\gd^{(t)}, 
\end{align}
for some \textit{sufficiently well-behaved} initial function $g^{(0)}$. Indeed, we will prove that $\cT$ induces a contraction mapping on the space $C_0(\R^{n-1},[1,n])$ and hence, converges to a unique fixed point. Throughout our discussion, $C_0(\R^{n-1},[1,n])$ is the space of continuous functions from $\R^{n-1}$ to $[1,n]$ embedded with the uniform norm: 
$$
\| g_1 - g_2 \| := \sup_{y \in \R^{n-1}} |g_1(y) - g_2(y)|.
$$

Once the fixed point function $g$ is found, $\cI h$ is derived from $g$ according to \eq{def-g} and then $h$ is derived from $\cI h$ as in \eq{eq-h}. However, we will need $h$ to be symmetric according to Definition\,\ref{def-sym}. If a function $g:\R^{n-1} \rightarrow \R$ leads to a symmetric $h:\R^n \rightarrow \R$ through the mentioned transformations, then $g$ is called a \emph{quasi-symmetric} function. Note that if $g$ is quasi-symmetric, then $\gd$ is also quasi-symmetric. Because the definition of $\cT$ suggests that it is independent of the labeling of indices $l \in [n] \setminus \{i\}$ and also there is a symmetry between $\cI h(y_{i \setminus k})$ and $y_{ji}$, for $j \in [n] \setminus \{k,i\}$ in \eq{def-g}. Therefore, we limit our attention to the set of quasi-symmetric functions. 

Moreover, in order to make sure that $g$ leads to a strictly decreasing $\cI h$ and consequently a threshold function $h$, we impose a stronger condition on $g$. The Lipschitz continuity is imposed on $g : \R^{n-1} \rightarrow \R$, where $\R^{n-1}$ is embedded with the $L_1$ norm. We show that the Lipschitz continuity with parameter $a_n$ is preserved through $\cT$ under certain conditions, i.e., we require that for any $y_{i\setminus k}, y'_{i\setminus k}  \in \R^{n-1}$:
\be{eq6}
|g(y_{i\setminus k}) - g(y'_{i\setminus k})| \leq a_n|x_i-x'_i| +a_n \sum_{j \in [n]\setminus \{k,i\}}| y_{ji} - y'_{ji} |.
\ee

Let $\cF$ denote the space of all Lipschitz continuous functions $f$ with parameter $a_n$ which are also quasi-symmetric. We aim at characterizing a condition on the noise variance parameters of the system such that the Lipschitz continuity is preserved through the operation $\cT$. In other words $\cT(\cF) \subseteq \cF$, where $\cT(\cF) = \left\{ \cT(g): g \in \cF \right\}$. 

The following lemma establishes a sufficient condition on the parameters of the system to preserve the Lipschitz continuity through the operation $\cT$. The Lipschitz continuity of $g$ is used to show Lipschitz continuity of $\cM_{\beps,l} g$, and consequently the Lipschitz continuity of $\cT g$, where the Lipschitz continuity of $\phi(.)$ together with triangle inequality are used. Furthermore, an upper bound on the Lipschitz constant of $\cM_{\beps,l} g$ and consequently on $\cT g$ are derived in terms of the Lipschitz constant of $g$ and other parameters of the system. That leads to a sufficient condition for not increasing the Lipschitz constant through the operation $\cT$, described in the following lemma. Let 
\be{def-e}
e_n := \max\{c_n a_{n-1}, d_n, c_n b_{n-1}\}.
\ee

\begin{lemma}
\label{lem2}
If 
\be{lem2-1}
\frac{(n-1)\bigl(e_n(na_n+(n-2)b_n)(b_n+2a_n)+b_n^2\bigr)}{a_n b_n^2} \leq \tau,
\ee
then Lipschitz continuity with parameter $a_n$ is preserved through the operation $\cT$.
\end{lemma}
\begin{proof}
Let $y, y' \in \R^{n-1}$. Then by definition of operation $\cT$,
\begin{align}
&\gd(y) - \gd(y') \nonumber\\
%\label{lem2-1}
&= \sum_{l \in [n]\setminus \{i\}} \int_{\R^{n-1}} \bigl(\varphi(\cM_{\beps,l} g(y)) - \varphi(\cM_{\beps,l} g(y'))\bigr)  f_{\beps} ({\beps}) d \beps.
\end{align}

Observe that $|\varphi(A) - \varphi(B)| \leq |A-B|$ for any $A,B \in \R$. Therefore,
\begin{equation}
\label{lem2-2}
\begin{split}
&|\gd(y) - \gd(y')| \\
&\leq  \sum_{l \in [n]\setminus \{i\}} \int_{\R^{n-1}} |\varphi(\cM_{\beps,l} g(y)) - \varphi(\cM_{\beps,l} g(y'))|  f_{\beps} ({\beps}) d \beps \\
& \leq   \sum_{l \in [n]\setminus \{i\}} \int_{\R^{n-1}} |\cM_{\beps,l} g(y) - \cM_{\beps,l} g(y')|  f_{\beps} ({\beps}) d \beps.
\end{split}
\end{equation}

Let $\delta = \|y-y'\|_1$. In the following series of inequalities, we only use the Lipschitz continuity of $g$ as given in \eq{eq6} and triangle inequality. For any $l \in [n]\setminus\{i\}$ and $\beps \in \R^{n-1}$,
\be{lem2-3}
\begin{split}
& b_n\tau \bigl |\cM_{\beps,l} g(y) - \cM_{\beps,l} g(y')\bigr | \\
& \leq |g(\overline{x_l}+\ep_l,\{\overline{y_{jl}}+\ep_{jl}\}_{j \in [n] \setminus \left\{i,l\right\}}) \\
&- g(\overline{x_l}'+\ep_l,\{\overline{y_{jl}}'+\ep_{jl}\}_{j \in [n] \setminus \left\{i,l\right\}})| + a_n |\overline{x_l} - \overline{x_l}'| \\
& + b_n\sum_{j \in \setminus [n] \{i,l\}} |\overline{y_{jl}} - \overline{y_{jl}}'| +b_n \delta \\
& \leq 2a_n |\overline{x_l} - \overline{x_l}'| + (a_n+b_n)\sum_{j \in [n] \setminus \{i,l\}} |\overline{y_{jl}} - \overline{y_{jl}}'| +b_n \delta \\
\end{split}
\ee
where $\overline{x_l}$ and $\overline{y_{jl}}$ are derived in~\eqref{eq:mean} with $y_{i \setminus k} = y$ and $y_{ki} = \frac{1}{b_n}g(y_{i\setminus k}) - \frac{a_n}{b_n}x_i - \sum_{j \in [n] \setminus \{k,i\}} y_{ji}$, according to Definition\,\ref{def1}. Also, $\overline{x_l}'$ and $\overline{y_{jl}}'$ are derived similarly. Then we have
\be{lem2-4}
\begin{split}
|\overline{x_l} - \overline{x_l}'|
&\leq e_n \delta + e_n |y_{ki} - y'_{ki}|\\
&\leq e_n(1+\frac{a_n}{b_n}) \delta +  \frac{e_n}{b_n}|g(y_{i\setminus k}) - g(y'_{i\setminus k})| \\
&\leq e_n(1+\frac{2a_n}{b_n}) \delta,
\end{split}
\ee
where in the last inequality we used the Lipschitz continuity of $g$. Similarly, we have
\be{lem2-5}
|\overline{y_{jl}} - \overline{y_{jl}}'| \leq e_n(1+\frac{2a_n}{b_n}) \delta.
\ee
By combining \eq{lem2-3}, \eq{lem2-4} and \eq{lem2-5} we have
\be{lem2-6}
\begin{split}
&b_n\tau \bigl |\cM_{\beps,l} g(y_{i\setminus k}) - \cM_{\beps,l} g(y'_{i\setminus k})\bigr |\\
&\leq\frac{1}{b_n}\bigl(e_n(na_n+(n-2)b_n)(b_n+2a_n)+b_n^2\bigr)\delta.
\end{split}
\ee
This together with \eq{lem2-2} imply that
\begin{align*}
&|\gd(y_{i\setminus k}) - \gd(y'_{i\setminus k})|\\
& \leq \frac{(n-1)\bigl(e_n(na_n+(n-2)b_n)(b_n+2a_n)+b_n^2\bigr)}{b_n^2 \tau} \delta.
\end{align*}
Therefore, a sufficient condition on $\gd$ for being Lipschitz continuous with parameter $a_n$ is that
$$
\frac{(n-1)\bigl(e_n(na_n+(n-2)b_n)(b_n+2a_n)+b_n^2\bigr)}{b_n^2 \tau} \leq a_n.
$$
which completes the proof. 
\end{proof}

Suppose that the parameters of the system satisfy \eqref{lem2-1}. First notice that since $\varphi(\alpha)\in[0,1]$, it follows that $\gd(y)\in[1,n]$ for all $y\in\R^{n-1}$. Therefore, one can view $\cT$ as a mapping $\cT:\cF\cap C_0(\R^{n-1},[0,n])\to \cF\cap C_0(\R^{n-1},[1,n])$. The following lemma builds upon \Lref{lem2} and exploits the Lipschitz continuity of $g$, which is preserved through the operation $\cT$, to show that $\cT$ is indeed a contraction mapping for a certain set of parameters.  
\begin{lemma}
\label{lem3}
For any two functions $g,h \in \cT(\cF)$,
$$
\| \cT g_1 - \cT g_2\| \leq \frac{(n-1)(e_n(na_n+(n-2)b_n)+ b_n)}{b_n^2 \tau}\|g_1-g_2\|. 
$$
\end{lemma}
\begin{proof}
For any $y \in \R^{n-1}$, 
\be{lem3-1}
\begin{split}
&\cT g_1(y) - \cT g_2(y) \\
& =  \sum_{l \in [n]\setminus \{i\}} \int_{\R^{n-1}} \bigl(\varphi(\cM_{\beps,l} g_1(y) - \varphi(\cM_{\beps,l} g_2(y))\bigr)  f_{\beps} ({\beps}) d \beps\\
&\leq \sum_{l \in [n]\setminus \{i\}} \int_{\R^{n-1}} |\cM_{\beps,l} g_1(y)) - \cM_{\beps,l} g_2(y)|  f_{\beps} ({\beps}) d \beps\\
&\leq  \sum_{l \in [n]\setminus \{i\}}  \sup_{\beps \in \R^{n-1}}|\cM_{\beps,l} g_1(y)) - \cM_{\beps,l} g_2(y')|.
\end{split}
\ee
For any $\beps \in \R^{n-1}$ and $l \in [n]\setminus \{i\}$, by definition of $\cM_{\beps,l} g_1$ and $\cM_{\beps,l} g_2$ in \eq{defM} and triangle inequality we have
\be{lem3-2}
\begin{split}
&b_n\tau |\cM_{\beps,l} g_1(y_{i\setminus k}) - (\cM_{\beps,l} g_2(y_{i\setminus k})|\\
\leq& |g_1(\overline{x^1_l}+\ep_l,\{\overline{y_{jl}^1}+\ep_{jl}\}_{j \in [n] \setminus \left\{i,l\right\}}) \\
&- g_2(\overline{x^2_l}+\ep_l,\{\overline{y^2_{jl}}+\ep_{jl}\}_{j \in [n] \setminus \left\{i,l\right\}})|\\
& + a_n|\overline{x^1_l}-\overline{x^2_l}|+ b_n \sum_{j \in [n] \setminus \{i,l\}} |\overline{y^1_{jl}}-\overline{y^2_{jl}}|,
\end{split}
\ee
where $\overline{x^m_l}$ and $\overline{y^m_{jl}}$, for $m=1,2$, are derived in~\eqref{eq:mean} with $y_{i\setminus k} = y$ and $y_{ki} = \frac{1}{b_n}g_m(y_{i\setminus k}) - \frac{a_n}{b_n}x_i - \sum_{j \in [n] \setminus \{k,i\}} y_{ji}$, according to Definition\,\ref{def1}. Note that, for $y,y' \in \R^{n-1}$, 
\be{lem3-3}
\begin{split}
|g_1(y) - g_2(y')| &\leq |g_1(y) - g_1(y')| + |g_1(y')-g_2(y')|\\
 &\leq a_n\|y-y'\|_1 + \|g_1 - g_2\|, 
\end{split}
\ee
where we used \Lref{lem2} for Lipschitz continuity of $g_1$ with parameter $a_n$ and the triangle inequality. Using \eq{lem3-3} in \eq{lem3-2} we get
\be{lem3-4}
\begin{split}
&b_n\tau |\cM_{\beps,l} g_1(y_{i\setminus k}) - (\cM_{\beps,l} g_2(y_{i\setminus k})|\\
&\leq 2a_n|\overline{x^1_l}-\overline{x^2_l}|+ (a_n+b_n)\hspace{-2mm} \sum_{j \in [n] \setminus \{i,l\}}\hspace{-2mm} |\overline{y^1_{jl}}-\overline{y^2_{jl}}| + \|g_1 - g_2\|.
\end{split}
\ee
Also,
\be{lem3-5}
|\overline{x^1_l} - \overline{x^2_l}| \leq \frac{e_n}{b_n} |g_1(y_{i\setminus k}) - g_2(y_{i\setminus k})| \leq \frac{e_n}{b_n}\|g_1-g_2\|,
\ee
and similarly,
\be{lem3-6}
|\overline{y^2_{jl}} - \overline{y^2_{jl}}| \leq \frac{e_n}{b_n}\|g_1-g_2\|.
\ee
Therefore, \eq{lem3-4}, \eq{lem3-5} and \eq{lem3-6} together imply that
\be{lem3-7}
\begin{split}
&b_n\tau |\cM_{\beps,l} g_1(y_{i\setminus k}) - (\cM_{\beps,l} g_2(y_{i\setminus k})| \nonumber\\
&\leq (\frac{na_n e_n}{b_n} +(n-2)e_n+ 1)\|g_1 - g_2\|.
\end{split}
\ee
This together with \eq{lem3-1} imply that
\begin{align*}
&\cT g_1(y) - \cT g_2(y)\\
& \leq \frac{(n-1)(e_n(na_n+(n-2)b_n)+ b_n)}{b_n^2 \tau} \|g_1 - g_2\|,
\end{align*}
which completes the proof of lemma.
\end{proof}

Let
\be{def-w}
\begin{split}
w_n = \max \big\{&\frac{(n-1)(e_n(na_n+(n-2)b_n)+ b_n)}{b_n^2}, \\
&\frac{(n-1)\bigl(e_n(na_n+(n-2)b_n)(b_n+2a_n)+b_n^2\bigr)}{a_n b_n^2} \big\}
\end{split}
\ee
If $w_n < \tau$, then the conditions of \Lref{lem2} and \Lref{lem3} are both satisfied. Hence, if $g$ is Lipschitz continuous with parameter $a_n$, then so is $\cT g$ by \Lref{lem2}. Then by \Lref{lem3}, we have
$$
\| \cT g_1 - \cT g_2\| \leq \frac{w_n}{\tau} \|g_1-g_2\|. 
$$
In other words, if $w_n < \tau$, then $\cT$ is a contraction mapping over $\cF$, with parameter $w_n/\tau < 1$. This leads to the following theorem. 
\begin{theorem}
\label{thm1}
Let $w_n < \tau$. Then the sequence of functions $\{g^{(t)}\}$ defined by \eqref{eqn:iter}, where $g^{(0)} \in \cF$, converges to a unique fixed point $g=\cT g$ in the space $\cF$ embedded with the $\mathcal{L}_\infty$ norm. 
\end{theorem}
\begin{proof}
Let $w_n < \tau$. Then by \Lref{lem3}, $\cT$ would be a contraction mapping over $\cF$ with $\mathcal{L}_\infty$ norm (which is a complete space) and hence, the result follows immediately by the Banach Fixed Point Theorem \cite{khalil2002nonlinear}. 
\end{proof}
The following theorem is the main result of this section. 
%Since the condition on preserving the Lipschitz continuity in \Lref{lem2} is stronger and by replacing $a = a_2 = \frac{\sigma^2+\tau^2}{2\sigma^2+\tau^2}$ from \Lref{itm:dist}, we arrive at the following theorem.
\begin{theorem}
\label{thm2}
If $w_n < \tau$, then there exists a continuous threshold function $h :\R^n \rightarrow \R$ that leads to a threshold policy equilibrium. Furthermore, $h(y)$ can be numerically approximated with arbitrarily enough precision, with respect to $\mathcal{L}_{\infty}$ norm.
\end{theorem}
\begin{proof}
The proof follows from \Tref{thm-main}, \Lref{lem2} and \Tref{thm1}. The continuity of $g$ (and hence, $h$) follows from the fact that if the above condition holds, then $\cT$ would be a contraction mapping from $\cF\cap C_0(\R^{n-1},[1,n])$ to itself and $\cF\cap C_0(\R^{n-1},[1,n])$ is a closed subset of $C_0(\R^{n-1},[1,n])$. Therefore, the iterations converge to a function in $\cF$ which is (Lipschitz) continuous. 
\end{proof}

\begin{figure*}
\begin{center}
\includegraphics[width = .42\linewidth]{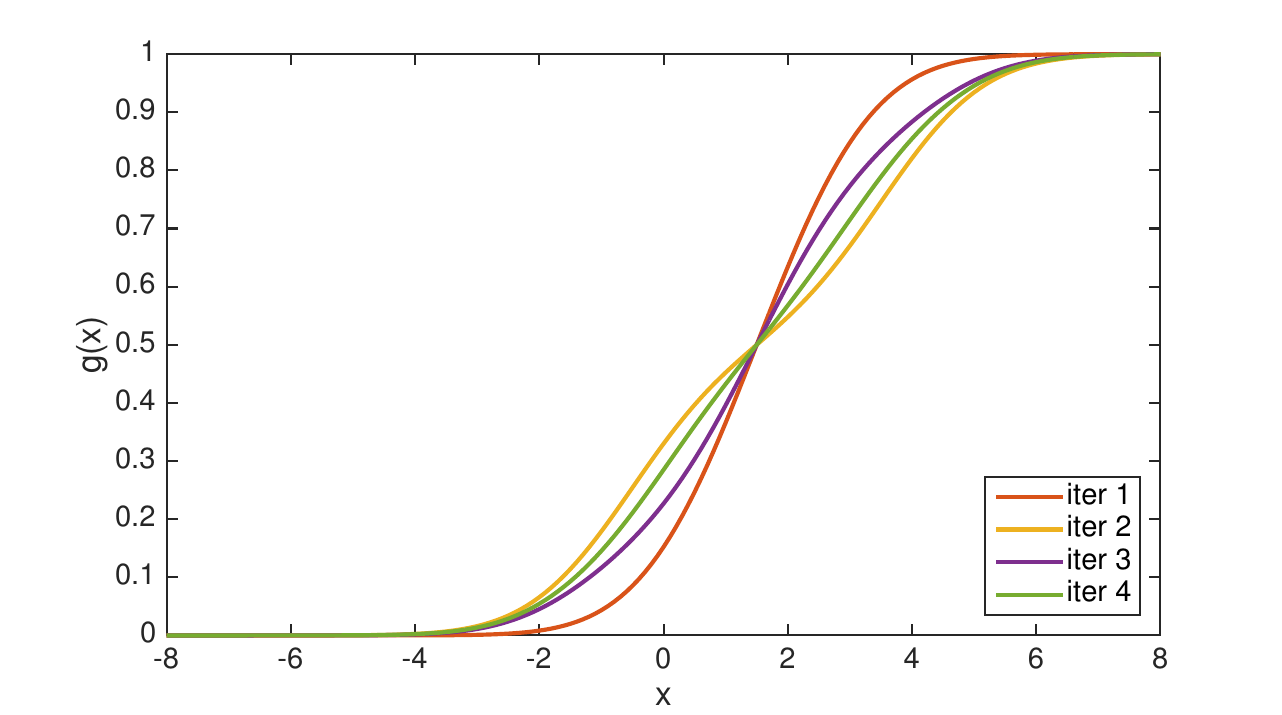}
\hspace{0.5in}
\includegraphics[width = .42\linewidth]{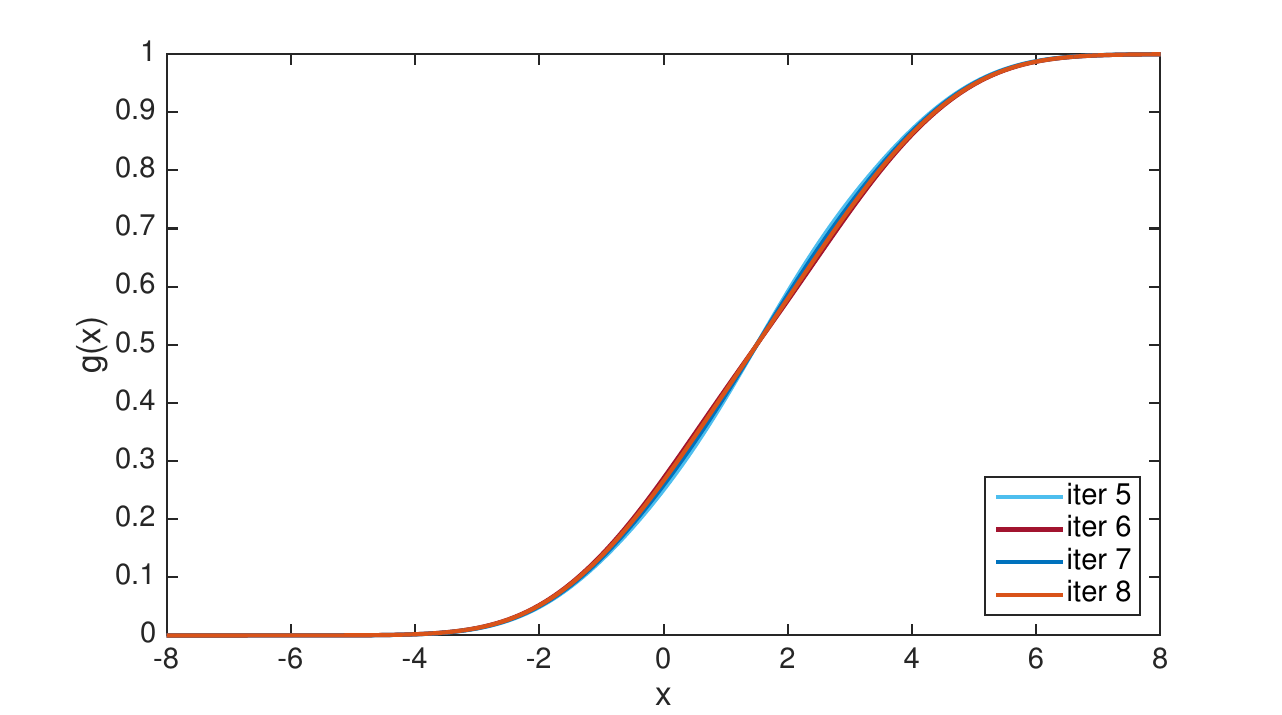}
\caption{The iterative solution $g(x)$ for $\sigma^2 = 1$ and $\tau^2 =9$.}
\end{center}
\label{fig:simul-1}
\end{figure*}

Regarding the conditions on the noise variance parameters it should be noted that for a fixed ratio of $r=\frac{\tau}{\sigma}$, the parameters $a_n, b_n, c_n, d_n$, introduced in \Lref{lemma:properties} and \Lref{itm:dist}, $e_n$, defined in \eq{def-e}, and $w_n$, defined in \eq{def-w}, are all also fixed. Therefore, one can fix $r$ while scaling $\tau$ and $\sigma$ with a large enough constant $c$ in order to have $w_n < \tau$. Roughly speaking, for a sufficiently noisy environment the condition of \Tref{thm2} holds. 

Let $\Theta$ denote the set of all $(\sigma^2,\tau^2) \in \R^{+2}$ for which the threshold policy equilibrium exists. By \Tref{thm2}, if $w_n < \tau$, a condition fully characterized by $\sigma$, $\tau$ and $n$, then $(\sigma^2,\tau^2) \in \Theta$. Therefore, $\Theta$ is an unbounded set, indeed in any direction, as discussed above. This completes the proof of \Tref{thrm:thm3}.
\noindent
{\bf Remark.} Let
$$
v_n = \frac{(n-1)(e_n(na_n+(n-2)b_n)+ b_n)}{b_n^2 \tau}.
$$
Then by \Lref{lem3} the Lipschitz constant of $\cT$ is upper bounded by $v_n$. Note that the condition $w_n < \tau$, the condition of \Tref{thm1} and \Tref{thm2}, also imply that $v_n < 1$. Therefore, the $\mathcal{L}_{\infty}$ distance between the function $h_m$ after $m$ iterations of the contraction mapping $\cT$ and the symmetric threshold function $h$, which is the fixed point of $\cT$, is $O(v_n^m)$. Given fixed noise parameters $\sigma$ and $\tau$, the parameter $v_n$ is an increasing function of $n$. This implies that the speed of convergence is decreasing with the increase in the number of agents. In other words, as $n$ increases, we will need more iterations to get the same level of precision for the approximation of the symmetric threshold function $h$.

\section{Numerical Analysis}
\label{numerical}

In this section, we consider the simple case of a global game with two agents, i.e., $n = 2$ for a numerical analysis. In this case, the independent noisy observations are $x_1$ and $x_2$. Furthermore, agent $1$ observes $y_1 = x_2 + \zeta_{21}$ and agent $2$ observes $y_2 = x_2 + \zeta_{12}$. In \Tref{thrm:thm2}, we proved that there is a coordinated equilibrium strategy of the form 
$$
s = (1_{h(x_1,y_1) \geq 0},1_{h(x_2,y_2)\geq 0}),
$$ 
for some continuous and strictly decreasing function $h:\R^2 \rightarrow \R$, under the condition provided in \Tref{thm2} which can be reformulated as 
\be{main-condition}
\frac{4a^2+3a-1}{a(1-a)} \leq \tau,
\ee
where $a = a_2 = \frac{\sigma^2 +\tau^2}{2\sigma^2+\tau^2}$. Furthermore, the convergence method presented in Section\,\ref{convergence} can be used to find $\cI h:\R\rightarrow \R$ such that 
$$
h(x,y) = \cI h(x) - y
$$ 
is a fixed point threshold function. 

While \eq{main-condition} is a sufficient condition to establish the convergence, it is not a necessary condition. In fact, the set $\Theta$ of all $(\sigma,\tau) \in \R^{+2}$ for which the threshold policy equilibrium exists may contain points $(\sigma,\tau)$ that do not satisfy \eq{main-condition}. This is shown for a simple numerical example. Let $\sigma = 1$ and $\tau = 3$. It can be observed that this certain choice of $(\sigma,\tau)$ does not satisfy \eq{main-condition}. For simplicity, the fixed point equation is solved for the function $g(x) = ax+b\cI h(x) - 1$, which is slightly modified comparing to \eq{def-g}. In the simulations, we demonstrate the convergence of solution $g(x)$ of \eq{eqn:iter} after only eight iterations, shown in Figure~2.

Next we evaluate the effect of number of agents on the convergence of the iterative operation to a fixed point solution $g(x)$. As briefly discussed in Section\,\ref{convergence}, increasing the number of agents slows down the convergence. Furthermore, given that the noise variance parameters are fixed, it is expected that convergence happens up to a certain $n_0$ and for $n > n_0$, the iteration does not converge. We have observed this through numerical simulations. In these simulations, we have limited our attention to functions $g(.) : \R^2 \rightarrow \R$ that are functions of agent's private observation as well as the aggregate sum of all the other observations. While our theoretical analysis in Section\,\ref{proof} does not guarantee that the fixed point solution is of this certain form, we have observed the convergence given this additional constraint. It is also possible that there exists a more strict set of conditions to guarantee the existence of a fixed point solution given this additional constraint. Let $\sigma = 10$ and $\tau=30$. In Figure\,3 the convergence of the iterative process for different values of the number of agents $n$ is shown. The initial function is $g^{(0)}=0$ on the entire $\R^2$. Then $g^{(t+1)} = \cT (g^{(t)})$, for $t\geq 0$, where $\cT$ is defined Definition\,\ref{def-T}. The distance between $g^{(t+1)}$ and $g^{(t)}$ is measured in terms of $\cL_2$ norm. It can be observed in Figure\,3 that as the number of agents $n$ increases, the convergence slows down. 

\begin{figure}
\centering
\label{plot3}
\includegraphics[width = .98\linewidth]{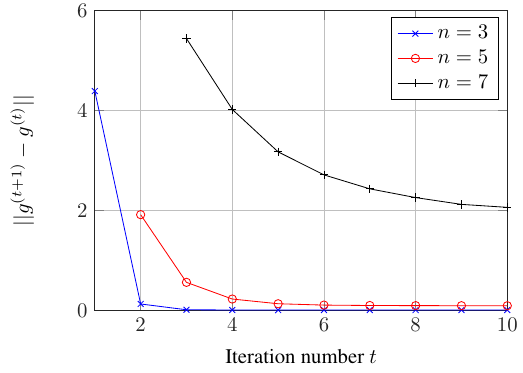}
\caption{The iterative convergence for different values of number of agents $n$}
\end{figure}

The effect of noise variance parameters on the convergence of the iterative operations is also evaluated. The number of agents is fixed to $n=5$. Let  
\be{ratio_sigma}
r = \frac{\sigma}{\tau} = \frac{1}{3}.
\ee
As mentioned in Section\,\ref{convergence}, the set of noise variance parameters $(\sigma,\tau)$ for which the iteration does not converge is bounded. Let $n=3$ and then the convergence for different values of $\sigma$ are evaluated, assuming the ratio $r$ is fixed as in \eq{ratio_sigma}. It is observed that for $\sigma = 1$, the convergence does not happen, however, it happens for $\sigma = 3,5,10$. Furthermore, the convergence for $\sigma=10$ is faster than the cases $\sigma = 3,5$, as expected. The results are shown in Figure\,4.

\begin{figure}
\centering
\label{plot4}
\includegraphics[width = .98\linewidth]{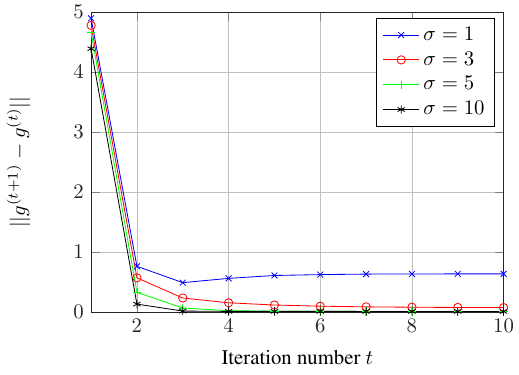}
\caption{The iterative convergence for different values of noise variances and $n=3$}
\end{figure}

\section{Conclusion}\label{sec:conclusion}
In this paper, we studied global games with noisy sharing of information and the existence and the structure of equilibria for those games. We showed that under some conditions on the noise parameters of the system, there exists a unique threshold-type equilibrium policy. To show this result, we formulated the existence of such an equilibrium as an existence of a solution to a functional fixed point equation. Using this formulation, we showed that for a set of noise parameters, the fixed point equation results in a contraction mapping on certain complete spaces and hence, we obtain a unique equilibrium based on threshold policies. Furthermore, such a contraction mapping would also provide us with a natural approximation mechanism. As discussed in Section\,\ref{numerical}, the conditions derived for existence of the threshold-type equilibrium policy may not be necessary. Therefore, characterizing the entire set of parameters for which the threshold-type equilibrium policy exists is still an open problem. 

\appendix
\label{appendix}

The next lemma determines the posterior of the unknown parameter given a bunch of independent Gaussian observations.
 \begin{lemma}
 \label{lem:parallel}
 Let $\theta$ be drawn from the non-informative uniform prior on $\R$. Further, let $\{\phi_i\}_{i \in [n]}$ denote noisy observations of $\theta$ with additive zero-mean independent Gaussian noises with variances $\{ \sigma^2_i \}_{i \in [n]}$. Let 
\begin{equation}
 \frac{1}{\sigma^2_\theta} := \sum_{i \in [n]} \frac{1}{\sigma^2_i}.
\label{eq:sigma-theta}
\end{equation}
Then, conditioned on $\{\phi_i\}_{i \in [n]}$, we can write $\theta$ as
 $$
 \theta = \sigma^2_\theta  \sum_{i \in [n]} \frac{\phi_i}{\sigma^2_i}+ z,
 $$
 where $z$ is $\cN(0,  \sigma^2_\theta)$ which is independent of $\{\phi_i\}_{i \in [n]}$. 
\end{lemma}
\begin{proof}
Notice that
\begin{align}
f(\theta \mid \{\phi_i\}_{i \in [n]} ) &= \frac{f(\{\phi_i\}_{i \in [n]} \mid \theta) }{\int_{-\infty}^\infty f(\{\phi_i\}_{i \in [n]} \mid \alpha) d \alpha}\nonumber\\
& = \frac{\exp\left( -\sum_{i \in [n]}  \frac{(\phi_i - \theta)^2}{2 \sigma^2_i} \right)}{  \int_{-\infty}^\infty   \exp\left( - \sum_{i \in [n]}  \frac{(\phi_i - \alpha)^2}{2 \sigma^2_i} \right) d\alpha}.
\label{eq:mean-var}
\end{align}
By inspecting the PDF in~\eqref{eq:mean-var}, it is clear that, conditioned on $\{\phi_i\}_{i \in [n]}$, $\theta$ follows a Gaussian distribution. It is therefore sufficient to calculate its mean and its variance.
Further notice that
\begin{align*}
\sum_{i \in [n]}  \frac{1}{\sigma^2_i} (\theta - \phi_i )^2 &=  \frac{1}{\sigma^2_\theta} \theta^2 - 2 \left(\sum_{i \in [n]}  \frac{\phi_i}{\sigma^2_i}\right) \theta + C_1\\
& = \frac{1}{\sigma^2_\theta} \left( \theta - \sigma^2_\theta  \sum_{i \in [n]}  \frac{\phi_i}{\sigma^2_i} \right)^2 + C_2,
\end{align*}
where $\sigma^2_\theta$ is defined in~\eqref{eq:sigma-theta}, and $C_1$ and $C_2$ are constants with respect to $\theta$ and only depends on $(\phi_1, \ldots, \phi_n).$
 Hence,
 $$
 f(\theta \mid \{\phi_i\}_{i \in [n]} ) \propto  \exp\left( -\frac{1}{2\sigma^2_\theta} \left(\theta - \sigma^2_\theta  \sum_{i \in [n]}  \frac{\phi_i}{\sigma^2_i} \right)^2  \right) ,
 $$
 which leads to the desired result.
\end{proof} 

\textit{Proof of \Lref{lemma:properties}.} Note that $(x_i, \{y_{ji} \}_{j \in [n] \setminus \{i\}})$ are zero-mean independent Gaussian observations of $\theta$. Further, $E[x_i^2] = \sigma^2$, whereas $E[y_{ji}^2] = \sigma^2 + \tau^2$.
	Therefore, the result follows by invoking Lemma~\ref{lem:parallel} and noticing the definition of $z_i$ in~\eqref{eqn:z}.

\textit{Proof of \Lref{itm:dist}.} 
	Since $\theta$ follows a non-informative uniform prior on $\R$, $x_k$ (or $y_{jk}$) also follows a uniform prior on $\R$. Therefore, $(x_i, \{y_{ji} \}_{j \in [n] \setminus \{i\}})$ can be viewed as  zero-mean independent Gaussian observations of $x_k$ given the dependence structure in Figure~\ref{fig:multi-agent}. Similarly, $(x_i, \{y_{ji} \}_{j \in [n] \setminus \{i\}})$ can be viewed as independent zero mean observations of 
	$y_{jk}$ for $j \neq k$.
	Hence, we can invoke $n$ instances of Lemma~\ref{lem:parallel}, where the state of the world is given by $x_k$, and $y_{jk}$ for $j \neq k$, respectively.

First, let us consider $y_{ik}$. Notice that  given $x_i$,  we see that $y_{ik}$ is independent of all other observations $\{y_{li} \}_{l \in [n] \setminus \{i\}} $, and it is a Gaussian observation of $x_i$ with noise variance $\tau^2$. 

Next, consider $x_k$. From the dependence structure in Figure~\ref{fig:multi-agent}, we see that $x_k$ depends on $(x_i, \{y_{ji} \}_{j \in [n] \setminus \{i,k\}})$ only through $\theta$. 
Hence, by invoking Lemma~\ref{lemma:properties} with $(n-1)$, we obtain the distribution of $\theta$ given $(x_i, \{y_{ji} \}_{j \in [n] \setminus \{i,k\}})$.
Furthermore, due to the non-informative prior, $\theta$ could be thought of as an independent Gaussian observation of $x_k$ with noise variance~$\sigma^2$. Thus, we invoke Lemma~\ref{lem:parallel} to obtain the distribution of~$x_k$.

The procedure for $y_{lk}$ with $l \neq i$ is similar and omitted for brevity.

	Finally, it is clear that, conditioned on $x_i$, $y_{ik}$ is independent of all other $\{x_j\}_{j \in [n]\setminus \{i\}}$ and $\{y_{jk} \}_{j \in [n] \setminus \{i\}}$. Hence, $ \ep_{ik}$ is independent of $(\ep_k, \{\ep_{lk}\}_{l \in [n] \setminus \{i,k\}  } )$ as desired.

\bibliographystyle{IEEEtran}
\bibliography{globalgames}

\end{document}